\documentclass{article}
\newtheorem{theorem}{Theorem}

\newtheorem{corollary}[theorem]{Corollary}

\newtheorem{proposition}[theorem]{Proposition}

\newenvironment{proof}[1][Proof]{\noindent\textbf{#1.} }{\ \rule{0.5em}{0.5em}}
\usepackage{amsmath}
\usepackage{amssymb}
\title{Einstein--Maxwell Sum Rules\\and the Persistence of the Smooth-Brane No-Go Theorem}
\author{Paulo Michel Longo Tavares da Silva\\
Faculty of Technology of Botucatu (FATEC Botucatu)\\S\~ao Paulo, Brazil}
\date{}

\begin{document}
\maketitle

\begin{abstract}
We investigate whether a minimally coupled Abelian Maxwell sector can relax the Gibbons--Kallosh--Linde (GKL) obstruction to smooth five-dimensional warped braneworlds. For a canonical scalar field coupled to a bulk Maxwell field, we derive the one-parameter family of global consistency conditions while keeping the tangential $F_{\mu\nu}$ and mixed $F_{\mu5}$ sectors separate. This decomposition is essential because the Maxwell stress tensor is not traceless in five dimensions. We show that electromagnetic terms formally modify the integrated GKL constraints, but this modification does not provide a consistent evasion of the no-go theorem under the standard assumptions. Exact four-dimensional Poincar\'e invariance excludes a nonzero deterministic classical Maxwell field from the background, since its field strength would select preferred four-dimensional directions. The Maxwell sector therefore vanishes in the symmetric background and the usual smooth-brane obstruction is recovered. We also identify which assumptions must be relaxed for a nontrivial electromagnetic contribution and distinguish this background result from local consistency conditions for bulk gauge-field localization.
\end{abstract}

\section{Introduction}
Warped braneworld models provide a geometrical framework in which a four-dimensional spacetime is embedded in a higher-dimensional bulk. The Randall--Sundrum constructions established the canonical five-dimensional warped setting and motivated a broad class of scalar--gravity and smooth domain-wall realizations \cite{Randall1999a,Randall1999b,DeWolfe2000}. In five dimensions, the consistency of a compact warped geometry is constrained not only by the local Einstein equations but also by integrated relations obtained by combining different traces of those equations. Gibbons, Kallosh and Linde (GKL) showed that these relations provide powerful global tests of braneworld constructions and, in particular, lead to an obstruction to smooth periodic domain-wall realizations in the simplest five-dimensional Einstein--scalar setting \cite{Gibbons2001}. Generalizations to arbitrary dimensions and to modified gravitational sectors have subsequently been considered \cite{Leblond2001,Hoff2011,daSilva2017,Yang2024}. In particular, the $f(R)$--Einstein--Palatini formulation provides a five-dimensional example in which the smooth-brane consistency condition is modified by the gravitational sector itself \cite{daSilva2017}.

Gauge fields introduce a distinct question. Their stress tensor carries anisotropic information and, in five dimensions, the Maxwell stress tensor is not traceless. Moreover, the problem of a Maxwell field as part of a higher-dimensional background must be distinguished from the localization of a four-dimensional gauge zero mode. The localization of gauge fields in warped backgrounds has a long history \cite{Oda2000,Guerrero2010}. More recently, consistency conditions obtained directly from the Einstein equations have shown that normalizability alone is insufficient and can impose strong restrictions on bulk vector zero modes \cite{Freitas2020,Alencar2025,Alencar2026End}. Nonlinear electromagnetic sources have also been considered in warped backgrounds using local sum-rule methods \cite{Alencar2026NED}.

The purpose of the present work is narrower and complementary. We derive the \emph{global} GKL consistency family for a five-dimensional Einstein--Maxwell--scalar system, keeping separately the sectors $F_{\mu\nu}$ and $F_{\mu5}$. This decomposition makes transparent which electromagnetic components enter each integrated condition. We then confront the global result with the tensorial requirements of an exactly Poincar\'e-invariant four-dimensional background. No claim of priority is required for the GKL framework itself; the emphasis here is on the explicit Einstein--Maxwell--scalar decomposition and on the distinction between a formal global contribution and a Maxwell background compatible with the assumed four-dimensional symmetry.

For a flat four-dimensional slice, the resulting family takes the form
\begin{equation}
\oint dy\,W^{\xi+1}\left[(3-\xi)\Phi'^2+2(\xi+1)V(\Phi)
+\frac{2-\xi}{4\pi}F_{\mu5}F^{\mu5}
+\frac{\xi-1}{8\pi}F_{\mu\nu}F^{\mu\nu}\right]=0.
\label{intro_sumrule}
\end{equation}
Two values of the parameter are especially informative. At $\xi=-1$, the scalar potential drops out and the Maxwell terms formally modify the integrated scalar obstruction. At $\xi=3$, the scalar-gradient contribution cancels while both electromagnetic sectors remain. These observations motivate the detailed derivation below. However, once exact four-dimensional Poincar\'e symmetry is imposed on the full classical background, one finds $F_{\mu\nu}=F_{\mu5}=0$; consequently the $\xi=-1$ relation returns to the standard positive scalar-gradient condition, $\oint dy\,\Phi'^2=0$. Thus the electromagnetic contribution changes the formal global identity without evading the smooth-brane no-go theorem within the standard symmetric setup.

\section{Five-dimensional Einstein--Maxwell--scalar system}
We consider the action
\begin{equation}
S=\int d^5x\sqrt{-g}\left[
\frac{R}{16\pi G_5}
-\frac{1}{16\pi}F_{AB}F^{AB}
-\frac12\partial_A\Phi\,\partial^A\Phi
-V(\Phi)
\right],
\label{action}
\end{equation}
where capital Latin indices $A,B,\ldots=0,1,2,3,5$, Greek indices $\mu,\nu,\ldots=0,1,2,3$, and $F_{AB}=\partial_A A_B-\partial_B A_A$. The field equations are
\begin{equation}
G_{AB}=8\pi G_5 T_{AB},
\end{equation}
\begin{equation}
\nabla_A F^{AB}=0,
\qquad
\Box_5\Phi-V_{,\Phi}=0,
\end{equation}
with
\begin{equation}
T_{AB}=T^{(\Phi)}_{AB}+T^{(F)}_{AB},
\end{equation}
\begin{equation}
T^{(\Phi)}_{AB}=\partial_A\Phi\partial_B\Phi
-g_{AB}\left(\frac12\partial_C\Phi\partial^C\Phi+V\right),
\label{scalarT}
\end{equation}
and
\begin{equation}
T^{(F)}_{AB}=\frac{1}{4\pi}\left(
F_{AC}F_B{}^C-\frac14g_{AB}F_{CD}F^{CD}
\right).
\label{maxwellT}
\end{equation}
In $D$ dimensions the Maxwell trace is proportional to $(4-D)F^2$; hence it is nonzero in five dimensions. Explicitly,
\begin{equation}
T^{(F)A}{}_A=-\frac{1}{16\pi}F_{AB}F^{AB}.
\end{equation}
This observation is important for the coefficients entering the sum rules.

\section{Warped geometry and the GKL global identity}
We specialize to one compact internal coordinate $y$ and write
\begin{equation}
ds^2=W^2(y)\,\bar g_{\mu\nu}(x)dx^\mu dx^\nu+dy^2,
\label{metric}
\end{equation}
where $W(y)>0$ is the warp factor. The scalar background is taken to depend only on the extra coordinate,
\begin{equation}
\Phi=\Phi(y).
\end{equation}
For a periodic compact internal space, the GKL identity may be written as \cite{Gibbons2001,Alencar2026End}
\begin{equation}
\oint dy\,W^n\left[T_\mu{}^\mu+(2n-4)T_5{}^5\right]
=\frac{1-n}{8\pi G_5}\,\bar R\oint dy\,W^{n-2},
\label{GKL}
\end{equation}
where $\bar R$ is the scalar curvature associated with $\bar g_{\mu\nu}$. For a flat four-dimensional slice, $\bar R=0$.

To make contact with the parameter used below, define
\begin{equation}
n=\xi+1.
\end{equation}
We shall derive the matter combination entering Eq.~(\ref{GKL}) rather than importing coefficients from a lower-dimensional Maxwell trace.

\section{Partial traces of the matter stress tensor}
\subsection{Scalar sector}
For $\Phi=\Phi(y)$, Eq.~(\ref{scalarT}) gives
\begin{align}
T_\mu^{(\Phi)\,\mu}&=-2\Phi'^2-4V,\\
T_5^{(\Phi)\,5}&=\frac12\Phi'^2-V.
\end{align}
Therefore
\begin{equation}
T_\mu^{(\Phi)\,\mu}+(2n-4)T_5^{(\Phi)\,5}
=(n-4)\Phi'^2-2nV.
\label{scalarcomb}
\end{equation}

\subsection{Maxwell sector}
Introduce the two invariants
\begin{equation}
X\equiv F_{\mu\nu}F^{\mu\nu},
\qquad
Y\equiv F_{\mu5}F^{\mu5},
\end{equation}
so that
\begin{equation}
F_{AB}F^{AB}=X+2Y.
\end{equation}
From Eq.~(\ref{maxwellT}), the four-dimensional partial trace is
\begin{equation}
T_\mu^{(F)\,\mu}=-\frac{Y}{4\pi},
\label{maxtrace4}
\end{equation}
and the internal component is
\begin{equation}
T_5^{(F)\,5}=\frac{Y}{8\pi}-\frac{X}{16\pi}.
\label{max55}
\end{equation}
Consequently,
\begin{equation}
T_\mu^{(F)\,\mu}+(2n-4)T_5^{(F)\,5}
=\frac{n-3}{4\pi}Y-\frac{n-2}{8\pi}X.
\label{maxcomb}
\end{equation}

\section{Einstein--Maxwell--scalar global sum rule}
For $\bar R=0$, substitution of Eqs.~(\ref{scalarcomb}) and (\ref{maxcomb}) into Eq.~(\ref{GKL}), followed by multiplication by $-1$ and use of $n=\xi+1$, gives
\begin{equation}
\boxed{\begin{aligned}
0=\oint dy\,W^{\xi+1}\Bigg[&
(3-\xi)\Phi'^2+2(\xi+1)V(\Phi)\\
&+\frac{2-\xi}{4\pi}F_{\mu5}F^{\mu5}
+\frac{\xi-1}{8\pi}F_{\mu\nu}F^{\mu\nu}
\Bigg].
\end{aligned}}
\label{mastersum}
\end{equation}
Equation~(\ref{mastersum}) is the central global consistency condition used in the remainder of this work. It is important that $F_{\mu\nu}F^{\mu\nu}$ is not positive definite in Lorentzian signature; hence sign arguments must not treat it as an ordinary Euclidean square.

\subsection{The case $\xi=-1$}
Setting $\xi=-1$ eliminates the potential term and yields
\begin{equation}
\boxed{\oint dy\left[
\Phi'^2+\frac{3}{16\pi}F_{\mu5}F^{\mu5}
-\frac{1}{16\pi}F_{\mu\nu}F^{\mu\nu}
\right]=0.}
\label{ximinusone}
\end{equation}
In the absence of a Maxwell background, Eq.~(\ref{ximinusone}) reduces to
\begin{equation}
\oint dy\,\Phi'^2=0,
\end{equation}
which, for a regular real scalar on a compact internal space, implies $\Phi'=0$. This is the familiar smooth-brane obstruction in the corresponding GKL setting \cite{Gibbons2001}. In contrast, Eq.~(\ref{ximinusone}) by itself does \emph{not} imply $\oint\Phi'^2=0$ when a Maxwell background is retained. The electromagnetic terms must first be constrained by the local field equations and by the symmetries of the background.

\subsection{The case $\xi=3$}
For $\xi=3$, the scalar-gradient term vanishes, but both electromagnetic sectors survive. Equation~(\ref{mastersum}) becomes
\begin{equation}
8\oint dy\,W^4V(\Phi)
-\frac{1}{4\pi}\oint dy\,W^4F_{\mu5}F^{\mu5}
+\frac{1}{4\pi}\oint dy\,W^4F_{\mu\nu}F^{\mu\nu}=0,
\end{equation}
or equivalently
\begin{equation}
\boxed{
32\pi\oint dy\,W^4V(\Phi)
+\oint dy\,W^4F_{\mu\nu}F^{\mu\nu}
=\oint dy\,W^4F_{\mu5}F^{\mu5}.}
\label{xithree}
\end{equation}
Thus the $\xi=3$ member isolates a balance among the weighted scalar potential, the tangential Maxwell invariant, and the mixed Maxwell invariant.

\section{Compatibility with exact four-dimensional Poincar\'e symmetry}
The global relations above are necessary integrated conditions. They do not by themselves guarantee that a proposed matter configuration is compatible with the full tensorial Einstein equations. We now impose the symmetry assumptions of the standard flat warped background,
\begin{equation}
\bar g_{\mu\nu}=\eta_{\mu\nu}.
\end{equation}
For the metric (\ref{metric}), the mixed Einstein tensor vanishes,
\begin{equation}
G_{\mu5}=0,
\end{equation}
and hence
\begin{equation}
T_{\mu5}=0.
\label{mixedconstraint}
\end{equation}
The scalar background contributes no mixed component. Maxwell theory gives
\begin{equation}
T^{(F)}_{\mu5}=\frac{1}{4\pi}F_{\mu\nu}F_5{}^\nu,
\end{equation}
so Eq.~(\ref{mixedconstraint}) requires
\begin{equation}
F_{\mu\nu}F_5{}^\nu=0.
\label{mixedmaxconstraint}
\end{equation}

Exact four-dimensional Poincar\'e invariance imposes a stronger requirement: the background stress tensor in the four-dimensional directions must be proportional to $g_{\mu\nu}$. The Maxwell contribution is
\begin{equation}
\begin{aligned}
T^{(F)}_{\mu\nu}=\frac{1}{4\pi}\Big[&
F_{\mu\rho}F_\nu{}^\rho+F_{\mu5}F_\nu{}^5\\
&-\frac14g_{\mu\nu}(X+2Y)\Big].
\end{aligned}
\label{maxmunu}
\end{equation}
A nonzero four-vector $v_\mu\equiv F_{\mu5}$ produces the rank-one tensor $v_\mu v_\nu$, selecting a preferred four-dimensional direction. Likewise, a nonzero classical two-form $F_{\mu\nu}$ defines electric/magnetic directions and its quadratic stress is not invariant under the full Lorentz group. Therefore a deterministic Abelian Maxwell field that belongs to the background is incompatible with exact four-dimensional Poincar\'e invariance unless its field strength vanishes.

\begin{proposition}\textbf{Maxwell obstruction with exact four-dimensional Poincar\'e symmetry.}\\
Consider the five-dimensional Einstein--Maxwell--scalar system (\ref{action}) on a smooth warped background
\[
ds^2=W^2(y)\eta_{\mu\nu}dx^\mu dx^\nu+dy^2,
\qquad \Phi=\Phi(y),
\]
with a classical minimally coupled Abelian Maxwell field included as part of the background. If the full background preserves the four-dimensional Poincar\'e group, then
\begin{equation}
F_{\mu5}=0,\qquad F_{\mu\nu}=0,
\end{equation}
and therefore $F_{AB}=0$.
\end{proposition}

\begin{proof}
For the warped metric above, four-dimensional Poincar\'e invariance requires the background energy--momentum tensor to contain no nonzero four-vector and no nontrivial Lorentz tensor other than the metric itself. First define
\begin{equation}
v_\mu\equiv F_{\mu5}.
\end{equation}
The Maxwell stress tensor contains the contribution
\begin{equation}
\frac{1}{4\pi}v_\mu v_\nu.
\end{equation}
If $v_\mu\neq0$, this rank-one tensor identifies a preferred four-dimensional direction. It cannot be proportional to the nondegenerate Lorentz metric $g_{\mu\nu}$ for all components. Full Poincar\'e invariance therefore implies
\begin{equation}
F_{\mu5}=0.
\end{equation}

With the mixed sector absent, the remaining Maxwell contribution in the brane directions is
\begin{equation}
T^{(F)}_{\mu\nu}=\frac{1}{4\pi}\left(
F_{\mu\rho}F_\nu{}^\rho-\frac14g_{\mu\nu}F_{\rho\sigma}F^{\rho\sigma}
\right).
\end{equation}
A nonzero deterministic antisymmetric two-form $F_{\mu\nu}$ decomposes, in any local Lorentz frame, into electric and magnetic spatial vectors. Such a configuration selects preferred directions and is not invariant under the full four-dimensional Lorentz group. Equivalently, there is no nonzero Lorentz-invariant antisymmetric rank-two tensor in four-dimensional Minkowski spacetime. Hence exact Poincar\'e invariance of the background requires
\begin{equation}
F_{\mu\nu}=0.
\end{equation}
Together these conditions give $F_{AB}=0$. The mixed Einstein equation $G_{\mu5}=8\pi G_5T_{\mu5}$ is then automatically satisfied. Notice that the conclusion follows from the symmetry of the complete classical background and is stronger than the integrated sum rule alone.
\end{proof}

\noindent
This proposition concerns a deterministic Maxwell field belonging to the background. It does not state that Maxwell perturbations or effective four-dimensional gauge excitations cannot be considered. Their localization and dynamical consistency constitute a different problem. Recent local analyses show that minimally coupled Maxwell zero modes face an independent no-go obstruction when gravitational consistency and localization are imposed simultaneously \cite{Alencar2025,Alencar2026End}.

\begin{corollary}\textbf{Persistence of the GKL smooth-brane obstruction.}
Under the hypotheses of the proposition, the $\xi=-1$ global sum rule becomes
\begin{equation}
\oint dy\,\Phi'^2=0.
\end{equation}
For a real smooth scalar on a compact internal space, $\Phi'^2\geq0$ pointwise, and therefore
\begin{equation}
\Phi'(y)=0
\end{equation}
everywhere. Consequently, a minimally coupled classical Maxwell background does not relax the standard GKL smooth-brane obstruction while exact four-dimensional Poincar\'e invariance, compactness and regularity are maintained.
\end{corollary}

\section{Global modification versus actual evasion of the no-go theorem}
It is useful to distinguish two logically different statements. Equation~(\ref{ximinusone}) shows that, before the complete background constraints are imposed, Maxwell theory \emph{modifies} the $\xi=-1$ global integral condition. In that limited sense the purely scalar GKL relation is relaxed at the level of the formal integrated identity. This observation, however, is not an evasion of the no-go theorem. An actual evasion requires a nonzero electromagnetic configuration that simultaneously satisfies the Maxwell equations, the complete Einstein equations, the global sum rules, the compactness and regularity assumptions, and the symmetries assigned to the background.

For the standard flat warped ansatz the proposition above excludes such a classical Maxwell background. Hence
\begin{equation}
\boxed{\text{formal global modification}\;\not\Rightarrow\;\text{consistent no-go evasion}.}
\end{equation}
This distinction also separates the present argument from recent local consistency theorems for localized bulk fields. The latter constrain the existence of consistent and normalizable gauge zero modes in warped backgrounds \cite{Alencar2025,Alencar2026End}; here the question is instead whether a nonzero Maxwell field can belong to the classical Poincar\'e-invariant background and alter the GKL global obstruction. The two analyses are complementary rather than interchangeable.

\section{What must be relaxed for a nontrivial electromagnetic contribution?}
Equations (\ref{ximinusone}) and (\ref{xithree}) remain useful because they identify the assumptions under which electromagnetic contributions can become relevant. At least one ingredient of the standard setup must be changed. Possibilities include a four-dimensional geometry that is not maximally symmetric, anisotropic or compact-object geometries capable of supporting electromagnetic stress, additional fields whose combined stress tensor restores the desired symmetry, nonminimal gauge couplings, nonlinear electrodynamics, noncompact internal spaces with appropriate boundary terms, or additional localized sources.

This observation also clarifies the relation with recent local analyses. Local consistency conditions can be more restrictive than a single integrated identity and have led to strong no-go results for minimally coupled Maxwell zero modes in warped braneworlds \cite{Alencar2025,Alencar2026End}. Conversely, nonlinear electrodynamics has been used to construct nontrivial warped compact-object embeddings when the four-dimensional geometry itself supports electromagnetic structure \cite{Alencar2026NED}. The global relations derived here should therefore be regarded as necessary integral constraints and as a diagnostic tool for identifying which electromagnetic sectors can enter a proposed braneworld background.

\section{Discussion}
The decomposition into $X=F_{\mu\nu}F^{\mu\nu}$ and $Y=F_{\mu5}F^{\mu5}$ exposes two features that are obscured if Maxwell terms are inserted into the scalar sum rule by analogy. First, the mixed sector carries the coefficient $(2-\xi)/(4\pi)$ and therefore survives at $\xi=3$. Second, the scalar-gradient term vanishes at $\xi=3$, leaving a balance among the weighted scalar potential and the two Maxwell invariants through Eq.~(\ref{xithree}). This relation is formally nontrivial even though an exactly Poincar\'e-invariant Maxwell background is excluded by the stronger local symmetry requirement.

For $\xi=-1$, the global identity alone does not establish the original scalar no-go statement in the presence of arbitrary electromagnetic fields. Only after imposing the background symmetry does the Maxwell field vanish and the standard positive scalar-gradient integral reappear. This distinction between an integrated necessary condition and the complete tensorial consistency problem is essential.

The analysis suggests two natural extensions. One is to replace the flat four-dimensional metric by a geometry that can consistently carry electromagnetic stress and then solve the global conditions together with the local Einstein--Maxwell equations. A second is to replace linear Maxwell theory by a nonlinear or nonminimally coupled gauge sector. Both extensions alter the stress-tensor traces and hence the global sum-rule coefficients.

\section{Conclusions}
We have derived the global braneworld sum-rule family for a five-dimensional Einstein--Maxwell--scalar system, retaining separately the tangential and mixed components of the Maxwell field strength. For a flat four-dimensional slice the result is Eq.~(\ref{mastersum}). The derivation depends crucially on the fact that Maxwell theory is not traceless in five dimensions.

Two members of the family are particularly transparent. At $\xi=-1$, the scalar potential disappears and Maxwell modifies the formal global condition through both $F_{\mu5}F^{\mu5}$ and $F_{\mu\nu}F^{\mu\nu}$. Therefore the scalar no-go conclusion cannot be inferred from this integrated equation alone. At $\xi=3$, only the scalar-gradient sector cancels identically; the weighted relation (\ref{xithree}) contains the scalar potential and both Maxwell invariants.

The central conclusion is therefore sharper than a formal modification of the integrated constraint. When the full symmetry assumptions of the standard flat warped ansatz are imposed, a nonzero deterministic classical Maxwell background selects preferred four-dimensional directions. Exact Poincar\'e invariance forces $F_{AB}=0$, and the $\xi=-1$ relation reduces to the standard GKL smooth-brane obstruction. In this precise sense the no-go theorem persists in the minimally coupled Einstein--Maxwell--scalar system. The electromagnetic terms in the global family remain useful because they identify exactly how the integrated obstruction changes once one relaxes Poincar\'e invariance, linear Maxwell dynamics, compactness, regularity, minimal coupling, or the assumed matter content. Thus the global sum rule serves both as a consistency test and as a map of the assumptions that a genuine electromagnetic evasion would have to abandon.

\end{document}